\documentclass[preprint]{elsarticle}
\usepackage{lineno}
\newcommand{\longversion}[1]{}

\usepackage{array,xspace,multirow,hhline,tikz,colortbl,tabularx,booktabs,fixltx2e,amsmath,amssymb,amsfonts,amsthm}
\usepackage{algorithm}
\usepackage{algorithmic}
\usepackage{verbatim,ifthen}
\usepackage{enumitem}
\usepackage{pifont}
\usepackage{ifthen}
\usepackage{calrsfs,mathrsfs}
\usepackage{bbding,pifont}
\usepackage{pgflibraryshapes}

	\usepackage{varioref}

	\usepackage{hyperref}

\definecolor{light-gray}{gray}{0.9}

	\newtheorem{lemma}{Lemma}%
	\newtheorem{theorem}{Theorem}%
	\newtheorem{corollary}{Corollary}%
	\newtheorem{example}{Example}
		\newtheorem{claim}{Claim}
			
	\usepackage{enumitem}
	\setenumerate[1]{label=\rm(\it{\roman{*}}\rm),ref=({\it\roman{*}}),leftmargin=*}
	\newlength{\wordlength}

	\newcommand{\midd}{\mathbin{:}}

	\newcommand{\pref}{\succ\xspace}

\begin{document}

	\title{Complexity of Manipulating Sequential Allocation}
    		\author{Haris Aziz}\ead{haris.aziz@nicta.com.au}
                \address{Data61 and University of New South Wales, Sydney, Australia}
                                		\author{Sylvain Bouveret} \ead{sylvain.bouveret@imag.fr}
                                                    \address{LIG - Grenoble INP, France}
            \author{J{\'{e}}r{\^{o}}me Lang}\ead{Jerome.Lang@irit.fr}
            \address{LAMSADE and Universit{\'e} Paris-Dauphine, France}
    		\author{Simon Mackenzie} \ead{simon.mackenzie@nicta.com.au}
                \address{Data61 and University of New South Wales, Sydney, Australia}

\begin{abstract}
	Sequential allocation is a simple allocation mechanism in which agents are given pre-specified turns and each agents gets the most preferred item that is still available. It has long been known that sequential allocation is not strategyproof. 

Bouveret and Lang (2014) presented a polynomial-time algorithm to compute a best response of an agent with respect to additively separable utilities and claimed that (1) their algorithm correctly finds a best response, and (2) each best response results in the same allocation for the manipulator. We show that both claims are false via an example. We then show that in fact the problem of computing a best response is NP-complete. 
On the other hand, the insights and results of Bouveret and Lang (2014) for the case of two agents still hold. 

\end{abstract}

\maketitle

% Bouveret and Lang ([2], Proposition 7) show that finding a manipulation for an n-agent picking sequence can be reduced to finding a manipulation for a 2-agent picking sequence.

\section{Introduction}

	A simple but popular mechanism to allocate
	indivisible items is {\em sequential allocation} 
	\citep{AWX15b,BoLa11a,BrSt79a,BrTa96a,KNW13a,KoCh71a,LeSt12a}. %ACMM05a,DeHi88a,Gard73b 
	In sequential allocation, a sequence specifies the turns of the agents. 
For example, for sequence 1212, agents 1 and 2 alternate with agent 1 taking the first turn. 	
	Agent take turns according to the sequence 
	and are given the 
	the most preferred item that has not yet been allocated. 
Sequential allocation is an ordinal mechanism since the outcome only depends on the ordinal preferences of agents over items. Nevertheless, it is a standard assumption in the literature that agents have underlying additive utilities for the items.
	
		 	It has long been known that sequential allocation is not strategyproof in particular when agents do not have consecutive turns. 
        This motivates the natural problem of computing best responses (also referred to as manipulations). 
			 \citet{KoCh71a} presented a polynomial-time algorithm to compute the optimal manipulation of an agent when there are two agents and the sequence is alternating (121212..). \citet{BoLa11a} initiated further work on manipulation of sequential allocation. They showed that it can be checked in polynomial time whether an agent can be allocated a certain subset of items. Since there can be exponential number of subsets, the result does not show that finding the optimal response is polynomial-time solvable as well.

\paragraph{Results}
	
We on focus on computing best responses (or manipulations) under sequential allocation. 
Recently, \citet{BoLa14a,BoLa14b} presented a polynomial-time algorithm to compute a best response of an agent and claimed that (1) their algorithm correctly finds an optimal response and (2) each best response results in the same allocation for the manipulator. The result has been cited in a number of works~\citep[see e.g.,][]{AGM+15c,HoLa15a,LaRe15a,Wals16a}. 
We first show that both claims are false by the help of an example. We then show that in fact the problem of computing a best response is NP-complete.
The result has some interesting consequences since many allocation rules are based on sequential allocation and for all such rules, there cannot be a general polynomial-time to manipulate the mechanism unless P=NP. 
Since manipulation by even one agent is NP-hard, the NP-hardness also implies a number of NP-hardness results by \citet{BoLa14a,BoLa14b} on \emph{coalitional} manipulation.

\section{Preliminaries}

We consider the setting in which we have $N=\{1,\ldots, n\}$ a set of agents, $O=\{o_1,\ldots, o_m\}$ a set of items, and the preference profile $\pref=(\pref_1,\ldots, \pref_n)$ specifies for each agent $i$ his complete, strict, and transitive preference $\pref_i$ over $O$. % Agents may be indifferent among items.
% We denote by $E_i^1, \ldots , E_i^{k_i}$ the $k_i$ equivalent classes of an agent $i\in N$. Those classes partition $O$ into $k_i$ sets of items such that agent $i$ is indifferent between two items belonging to the same class, and he strictly prefers an item of $E_i^k$ to an item of $E_i^l$ whenever $k>l$.

% {An assignment $x=(x_1, \ldots , x_n)$ is a partition of $O$ into $n$ subsets,
% where $x_{i}$ is the bundle assigned to agent $i$. We denote by $\mathcal{X}$ the set of all possible assignments.}
Each agent may additionally express a cardinal utility function $u_i$ consistent with $\pref_i$: $u_i(o)> u_i(o') \text{ iff } o\pref_i o'.$ We will assume that each item is positively valued, i.e, $u_i(o)>0$ for all $i\in N$ and $o\in O$. The set of all utility functions consistent with $\pref_i$ is denoted by $\mathcal{U}(\pref_i)$. We will denote by $\mathcal{U}(\pref)$ the set of all utility profiles $u=(u_1,\ldots, u_n)$ such that $u_i\in \mathcal{U}(\pref_i)$ for each $i\in N$. 
When we consider agents' valuations according to their cardinal utilities, then we will assume additivity, that is $u_i(O')=\sum_{o\in O'}u_i(o)$ for each $i\in N$ and $O'\subseteq O$.

At times, we will present an assignment in matrix form.
			An \emph{assignment} is an $n\times m$ matrix $[p(i)(o_j)]_{\substack{1\leq i\leq n, 1\leq j\leq m}}$ such that for all $i\in N$, and $o_j\in O$, $p(i)(h_j)\in \{0,1\}$;  and for all $j\in \{1,\ldots, m\}$, $\sum_{i\in N}p(i)(o_j)= 1$. %$\sum_{h_j\in H}p_{ij}=\sum_{h\in H}w(h)/n$;
An agent $i$ gets item $o_j$ if and only if $p(i)(o_j)= 1$.
 Each row $p(i)=(p(i)(h_1),\ldots, p(i)(o_m))$ represents the \emph{allocation} of agent $i$.

We say that utilities are \emph{lexicographic} if for each agent $i\in N$, {$u_i(o)>\sum_{o'\prec_i o}u_i(o')$}. By $S\pref_i T$, we will mean $u_i(S)\geq u_i(T)$.

	\begin{example} Consider the setting in which $N=\{1,2\}$, $O=\{o_1,o_2,o_3,o_4\}$, the preferences of agents are
		\begin{align*}
			1:\quad o_1, o_2, o_3, o_4\\
			2:\quad o_1, o_3, o_2, o_4
			\end{align*}
		Then for the policy $1221$, agent $1$ gets $\{o_1,o_4\}$ while $2$ gets $\{o_2,o_3\}$. The assignment resulting from sequential allocation (SA) can be represented as follows.
		
		\[
		SA(\pref_1,\pref_2) = \begin{pmatrix}
			1&0&0&1\\
			0&1&1&0
		\end{pmatrix}.
		\] 

		\end{example}

\section{Computing a Best Response: Case of Three or More Agents}

In the next example, we highlight that for $n\geq 3$, the best response algorithm of \citet{BoLa14a,BoLa14b} does not work and that an optimal manipulation may not result in a unique allocation.

\begin{example}
%Now this statement is certainly true for two agents but not true for 3 agents.

Observe the following preference profile. 
\begin{align*}
1:&\quad a, b, c, d\\
2:&\quad c, d, a, b\\
3:&\quad a, b, c, d
\end{align*}

Let the sequence be 1231.
Now according to the best response algorithm of \citet{BoLa14a}, the best response is one in which agent $1$ gets $\{a,d\}$ which can even be achieved by the truthful report. The reason the algorithm of \citet{BoLa14a,BoLa14b} returns the truthful report as the best response is because it first construct allocation $\{a\}$ and checks that it is achievable; then it checks whether $\{a,b\}$ is achievable and finds that it is not achievable; then it checks whether $\{a,c\}$ is achievable and finds that it is not achievable, and finally the algorithm terminates   when it is found that $\{a,d\}$ is achievable for agent 1 via the truthful report.

%So if 1 reports truthfully, he gets $\{a,d\}$. 
Let us assume that agent $1$ misreports so that the reported preference profile is as follows: %$1':\quad c,b,a,d$
\begin{align*}
1':&\quad c,b,a,d\\
2:&\quad c, d, a, b\\
3:&\quad a, b, c, d
\end{align*}

Under the misreport, agent $1$ gets $\{b,c\}$.
Agent 1 cannot get $\{a,b\}$ or $\{a,c\}$ which are not achievable. But he can get $\{a,d\}$ or $\{b,c\}$ depending on what he reports. Now which one is better clearly depends on the actual utilities and not just on the ordinal preferences. For example if the utilities are 3.1, 3, 2 and 1 then $\{b,c\}$ is preferred over $\{a,d\}$ and hence $\{a,d\}$ is not the best possible achievable allocation for agent $1$ which means that

Moreover, if the utilities are $4,3,2,1$, then agent $1$ is completely indifferent between the allocations $\{a,d\}$ and $\{b,c\}$. This proves that there may be best responses that do not yield a unique allocation for the manipulator when the number of non-manipulators is 2 or more (i.e., when the total number of agents is 3 or more).
\end{example}

The example above simply shows that algorithm of \citet{BoLa14a,BoLa14b} does not necessarily compute a best response that gives maximum utility to the agent. It does not settle the complexity of computing a best response. Next, we show that the problem is NP-hard. The reduction involves a similar high-level idea as that of the result by \citet{AGM+15c} that manipulating the probabilistic serial (PS) mechanism is NP-hard.
However, the reduction requires new gadgets.
Also note that the NP-hardness result for the PS mechanism does not directly imply a similar result for sequential allocation. Similarly, NP-hardness to manipulate sequential allocation does not imply NP-hardness to manipulate the PS mechanism. 

% So what are the consequences of this bug? It turns out that the complexity of an optimal manipulation is open.
% We conjecture that the complexity is NP-hard.

% \begin{conjecture}
% 	Computing a best response for the sequential allocation mechanism is NP-hard. %Also, verifying a pure Nash equilibrium is coNP-hard.
% 	\end{conjecture}
% The proof idea should be similar to the NP-hardness to manipulate SA.
% Simon thinks he has a proof for the reduction:)

\begin{theorem}\label{th:main}
	Computing a best response for the sequential allocation mechanism is NP-complete. %Also, verifying a pure Nash equilibrium is coNP-hard.
	\end{theorem}
	\begin{proof}

		To show hardness, we prove that the following problem (\textsc{Best Response}) is NP-complete: given an assignment setting and a utility function $u:O\rightarrow \mathbb{N}$ specifying the utility of each item for the manipulator (agent $1$) and a target utility $T$, can the manipulator specify preferences such that the utility for his allocation under the sequential allocation rule is at least $T$? 
The problem \textsc{Best Response} is clearly in NP. The outcome with respect to the reported preference can be computed by simulating sequential allocation. The utility achieved by the agent can be computed by adding the utility of the items allocated to the agent.

		We reduce from a restricted NP-complete version of 3SAT where each literal appears exactly twice in the formula. 
        The problem remains NP-complete~\citep[Lemma 1, ][]{BCF+15a}.
        Given such a 3SAT instance $F=(X,C)$ where $X = \{x_1,\dots,x_{|X|}\}$ is the set of variables and $C$ the set of clauses, we build an instance of \textsc{Best Response} where the manipulator can obtain utility $\geq T$ if and only if the formula is satisfiable. 
		We will denote by $L$ the set of literals. 
		
The set of agents is composed of 
	\begin{itemize}
		\item Agent 1 the manipulator/responder; and
		\item agents $a_{x_i}^1$ and $a_{x_i}^2$ for each literal $x_i$. 
	\end{itemize}
        % \fixme{Some inconsistency here: agent 2 appears in the set below but is not described
%           in the list of agents above. It seems to me that this second agent is not necessary
%           (see the second fixme at the end of the proof).}

		To summarize, the agent set is: 
	\[N=\{1\}\cup \{a_{x_i}^1,a_{x_i}^2\midd x_i \in L\}.\]

	The set of items $O$ is as follows:
	\begin{itemize}
		\item Clause items $o_c^1, o_c^2,o_c^3$ for each clause $c$;
		\item Choice items $o_{x_i}^1, o_{x_i}^2$ for each literal $x_i$;
		\item Consistency items $h_{x_i}^1,h_{x_i}^2,h_{x_i}^3$ for each literal $x_i$;
		\item Dummy items $d_{x_i}^{11}, d_{x_i}^{12}, d_{x_i}^{21}, d_{x_i}^{22}$ for each literal $x_i$. 
	\end{itemize}	
	To summarize, the item set is: 
	\[O=\{o_c^1, o_c^2,o_c^3 \midd c\in C\} \cup \{o_{x_i}^1, o_{x_i}^2, h_{x_i}^1,h_{x_i}^2,h_{x_i}^3, d_{x_i}^{11}, d_{x_i}^{12}, d_{x_i}^{21}, d_{x_i}^{22}\midd x_i\in L\}\]
	
	We view the sequential allocation process as follows. The preferences are built in a way so that the agents go through $|X|$ choice rounds corresponding to variables $x_1,\ldots, x_{|X|}$ and then $|C|$ clause rounds corresponding to $c_1,\ldots, c_{|C|}$ with one final round called the collection round in wich

\paragraph{High-level Idea} The items that agent $1$ potentially gets in each round have considerably more utility than items he is supposed to get in latter rounds. If agent $1$ does not get the items systematically, they will be taken by other agents and and then agent $1$ will not be able to make up for the loss of not getting those items first. This ensures that in each round agent $1$ makes a choice between the items corresponding to the positive and negative literals of a variable.
There is negligible difference between the utility of the items corresponding to the literal and its negation (for example $o_x^1$ and $o_{\neg x}^1$) so what is important is that agent chooses one of the items corresponding to the literals. If agent $1$ makes the correct choices, then it will ensure that it gets a most preferred items corresponding to each of the clause items.  	
We want to show that there is a satisfiable assignment if and only if agent $1$ gets utility $T$. This is only possible if agent $1$ sets the choice variables in a consistent way and manages to get the most preferred clause item for each clause which is only possible if each clause is set to be true which in turn is only possible if the agents corresponding to the negation of the literal in the clause do not get the clause items. The way the reduction works is also illustrated in Example~\ref{example:sat} right after the proof.
	
% Agent $1$ is required to choose a literal $x_i$ or $\neg x_i$ in choice round $i$. Since agent $1$ has to make a choice, if he picks $o_{x_i}$, then this means that agent $a_{x_i}$ cannot pick $o_{x_i}$ in that round and there picks the next most preferred item.This means that $a_{x_i}$ will be able to get an item corresponding to the clause it is in. We view $1$ picking $o_{x_i}$ as setting $x_i$ false. 

	\paragraph{Choice Round}

	\begin{table}[h!]
		\centering
			\scalebox{0.63}{
	\begin{tabular}{l|ccccc|ccccc|cccccc}%{| c | c | c | c | c | c | c | c | c |c | c | c |}
	  \hline
	\textbf{Stage}&  1 & 2 &3 &4&5&6&7&8&9&10&11&12&13&14&15&16\\
	  \hline 
	\textbf{Agent }&1& $a_{\neg x_i}^1$&$a_{\neg x_i}^2$ &$a_{x_i}^1$&$a_{ x_i}^2$&1&$a_{\neg x_i}^1$&$a_{\neg x_i}^2$&$a_{x_i}^1$&$a_{x_i}^2$&$a_{\neg x_i}^1$&$a_{\neg x_i}^2$&1&$a_{x_i}^1$&$a_{x_i}^2$&1  \\
	 \hline
	\textbf{Item picked}& $o_{\neg x_i}^1$& $o_{x_i}^1$&$d_{x_i}^{21}$ &$d_{\neg x_i}^{11}$&$d_{\neg x_i}^{21}$&$o_{\neg x_i}^2$&$d_{x_i}^{11}$&$o_{x_i}^2$&$h_{\neg x_i}^{1}$&$h_{\neg x_i}^{2}$&$d_{x_i}^{12}$&$d_{x_i}^{22}$&$h_{\neg x_i}^{3}$&$d_{\neg x_i}^{12}$&$d_{\neg x_i}^{22}$&$h_{x_i}^1$  \\
	  \hline
	\end{tabular}
	}
	\caption{Choice round for variable $x_i$ in which agent $1$ makes consistent choice $o_{\neg x_i}^1$ and $o_{\neg x_i}^2$ so that variable $x_i$ is set to true. Agents $a_{\neg x_i}^1$ and $a_{\neg x_i}^2$ next focus on items $h_{x_i}^2$ and $h_{x_i}^3$ before turning their attention to the clause items. On the other hand, $a_{x_i}^1$ and $a_{x_i}^2$ already want to get clause items. In this way the literal that is set false, their corresponding agents are quicker to get their clause items.}
	\label{table:consistent-true-new}
	\end{table}

	\begin{table}[h!]
		\centering
			\scalebox{0.63}{
	\begin{tabular}{l|ccccc|ccccc|cccccc}%{| c | c | c | c | c | c | c | c | c |c | c | c |}
	  \hline
	\textbf{Stage}&  1 & 2 &3 &4&5&6&7&8&9&10&11&12&13&14&15&16\\
	  \hline 
	\textbf{Agent }&1& $a_{\neg x_i}^1$&$a_{\neg x_i}^2$ &$a_{x_i}^1$&$a_{ x_i}^2$&1&$a_{\neg x_i}^1$&$a_{\neg x_i}^2$&$a_{x_i}^1$&$a_{x_i}^2$&$a_{\neg x_i}^1$&$a_{\neg x_i}^2$&1&$a_{x_i}^1$&$a_{x_i}^2$&1  \\
	 \hline
	\textbf{Item picked}& $o_{x_i}^1$& $d_{x_i}^{11}$&$d_{x_i}^{21}$ &$o_{\neg x_i}^1$&$d_{\neg  x_i}^{21}$&$o_{x_i}^2$&$d_{x_i}^{12}$&$d_{x_i}^{22}$&$d_{\neg  x_i}^{11}$&$o_{\neg x_i}^2$&$h_{x_i}^1$&$h_{x_i}^2$&$h_{\neg x_i}^1$&$h_{\neg  x_i}^{2}$&$h_{\neg  x_i}^{3}$&$h_{x_i}^3$  \\
	  \hline
	\end{tabular}
	}
	\caption{Choice round for variable $x_i$ in which agent $1$ makes consistent choice $o_{x_i}^1$ and $o_{x_i}^2$ so that variable $x_i$ is set to false. Agents $a_{x_i}^1$ and $a_{x_i}^2$ next focus on items $d_{\neg x_i}^{12}$ and $d_{\neg x_i}^{22}$  before turning their attention to the clause items. On the other hand, $a_{\neg x_i}^1$ and $a_{\neg x_i}^2$ already want to get clause items. In this way the literal that is set false, their corresponding agents are quicker to get their clause items.}
	\label{table:consistent-false-new}
	\end{table}

    \begin{table}[h!]
        \centering
            \scalebox{0.63}{
    \begin{tabular}{l|ccccc|ccccc|cccccc}%{| c | c | c | c | c | c | c | c | c |c | c | c |}
      \hline
    \textbf{Stage}&  1 & 2 &3 &4&5&6&7&8&9&10&11&12&13&14&15&16\\
      \hline
    \textbf{Agent }&1& $a_{\neg x_i}^1$&$a_{\neg x_i}^2$ &$a_{x_i}^1$&$a_{ x_i}^2$&1&$a_{\neg x_i}^1$&$a_{\neg x_i}^2$&$a_{x_i}^1$&$a_{x_i}^2$&$a_{\neg x_i}^1$&$a_{\neg x_i}^2$&1&$a_{x_i}^1$&$a_{x_i}^2$&1  \\
     \hline
    \textbf{Item picked}& $o_{x_i}^1$& $d_{x_i}^{11}$&$d_{x_i}^{21}$ &$o_{\neg x_i}^1$&$d_{\neg  x_i}^{21}$&$o_{\neg x_i}^2$&$d_{ x_i}^{12}$&$o_{x_1}^2$&$d_{\neg x_1}^{11}$&$h_{\neg x_i}^{1}$&$h_{x_i}^1$&$d_{ x_i}^{22}$&$h_{\neg x_i}^2$&$h_{\neg x_i}^{2}$&$h_{\neg x_i}^{3}$&$h_{x_i}^2$  \\
      \hline
    \end{tabular}
    }
    \caption{Choice round for variable $x_i$ in which agent $1$ makes inconsistent choice $o_{x_i}^1$ and $o_{\neg x_i}^2$. As a result of making an inconsistent choice, agent 1 does not get $h_{\neg x_i}^1$ or $h_{x_i}^1$.}
    \label{table:inconsistent-false-new}
    \end{table}

	\begin{table}[h!]
		\centering
			\scalebox{0.63}{
	\begin{tabular}{l|ccccc|ccccc|cccccc}%{| c | c | c | c | c | c | c | c | c |c | c | c |}
	  \hline
	\textbf{Stage}&  1 & 2 &3 &4&5&6&7&8&9&10&11&12&13&14&15&16\\
	  \hline 
	\textbf{Agent }&1& $a_{\neg x_i}^1$&$a_{\neg x_i}^2$ &$a_{x_i}^1$&$a_{ x_i}^2$&1&$a_{\neg x_i}^1$&$a_{\neg x_i}^2$&$a_{x_i}^1$&$a_{x_i}^2$&$a_{\neg x_i}^1$&$a_{\neg x_i}^2$&1&$a_{x_i}^1$&$a_{x_i}^2$&1  \\
	 \hline
	\textbf{Item picked}& $o_{\neg x_i}^1$& $o_{x_i}^1$&$d_{x_i}^{21}$ &$d_{\neg x_i}^{11}$&$d_{\neg x_i}^{21}$&$o_{x_i}^2$&$d_{x_i}^{11}$&$d_{x_i}^{22}$&$h_{\neg x_i}^{1}$&$o_{\neg x_i}^{2}$&$d_{x_i}^{12}$&$h_{x_i}^{1}$&$h_{\neg x_i}^{2}$&$d_{\neg x_i}^{12}$&$d_{\neg x_i}^{22}$&$h_{x_i}^2$  \\
	  \hline
	\end{tabular}
	}
	\caption{Choice round for variable $x_i$ in which agent $1$ makes inconsistent choice $o_{\neg x_i}^1$ and $o_{x_i}^2$ }
	\label{table:inconsistent-true-new}
	\end{table}

	This choice round is composed of two back to back sub rounds in which agent $1$ has to choose a literal corresponding to a variable and then do this again.

	Let us consider the truth assignment corresponding to the choices agent $1$ makes. A literal $x_i$ is considered to be true iff agent $1$ picks $ o_{\neg x_i}^1$ and $ o_{\neg x_i}^2$  in choice round $i$. 
     
The sequence in choice round $i$ for variable $x_i$ is as follows

    \[1, a_{\neg x_i}^1, a_{\neg x_i}^2, a_{x_i}^1, a_{ x_i}^2, 1, a_{\neg x_i}^1, a_{\neg x_i}^2, a_{x_i}^1, a_{x_i}^2, a_{\neg x_i}^1, a_{\neg x_i}^2,1, a_{x_i}^1, a_{x_i}^2, 1 \]

	The preferences relevant for the choice round are as follows. Each literal agent likes items corresponding to the negation of the literal. 
    
    For each variable $x_i$, agent $1$ has the following preferences
	
		\begin{align*}
			1:&\quad o_{x_i}^1, o_{\neg x_i}^1, o_{x_i}^2, o_{\neg x_i}^2, h_{\neg x_i}^1,h_{\neg x_i}^2,h_{\neg x_i}^3,h_{x_i}^1,h_{x_i}^2,h_{x_i}^3, 
			\end{align*}
            
            For each variable $x_i$, the preferences of the related agents are as follows:
    		\begin{align*}
    			a_{\neg x_i}^1:&\quad o_{x_i}^1, d_{x_i}^{11}, d_{x_i}^{12},o_{x_i}^2, h_{x_i}^1,  h_{x_i}^2,h_{x_i}^3 \\
    			a_{\neg x_i}^2:&\quad d_{x_i}^{21}, o_{x_i}^1, o_{x_i}^2,{d_{x_i}^{22}}, h_{x_i}^1,  h_{x_i}^2,h_{x_i}^3\\
    				\midrule
        			a_{x_i}^1:&\quad o_{\neg x_i}^1, d_{\neg x_i}^{11}, h_{\neg x_i}^1,o_{\neg x_i}^2, h_{\neg x_i}^2,h_{\neg x_i}^3,d_{\neg x_i}^{12} \\
        			a_{x_i}^2:&\quad d_{\neg x_i}^{21}, o_{\neg x_i}^1, o_{\neg x_i}^2, h_{\neg x_i}^1,h_{\neg x_i}^2,h_{\neg x_i}^3,d_{\neg x_i}^{22}
    			\end{align*}
			
			Any items that are not in the preference list of an agent are considered to be far down in the preference list. Note the difference in the preferences of the negative versus the positive literals: the positive literal agents have a dummy item as the least preferred item relevant to the picking in the choice round whereas the negative literal agents have the consistency items as the least preferred items relevant to the picking in the choice round.

The preferences agents in the choice round are made in such a way to that agent $1$ is compelled to make consistent choice so that 
it not only gets literal items corresponding to its choice but also one if its most two preferred consistency items{.}
If agent  makes a consistent choice $o_{x_i}^1$ and $o_{x_i}^2$, or $\neg o_{x_i}^1$ and $\neg o_{x_i}^2$, then $1$ gets either $\{h_{x_i}^1,h_{\neg x_i}^3\}$ or $\{h_{\neg x_i}^1,h_{x_i}^3\}$ for that round. This scenario is captured in  Table~\ref{table:consistent-true-new} and Table~\ref{table:consistent-false-new}.
If agent $1$ makes an inconsistent choice $o_{x_i}^1$ and $\neg o_{x_i}^2$, or $\neg o_{x_i}^1$ and $o_{x_i}^2$, then $1$ gets $h_{x_i}^2$ and $h_{\neg x_i}^2$. This scenario is captured in Table~\ref{table:inconsistent-false-new} and Table~\ref{table:inconsistent-true-new}.

		\paragraph{Clause round}
	
		The sequence in clause round corresponding to clause $c=(x_{i}\vee \neg x_{j}\vee \neg x_{k})$ is as follows. 

		 \[{a_{\neg x_i}}, a_{x_j},a_{x_k},---,1\] 
         
         For each literal $x$ in the clause $c$, there is an agent $a_{\neg x}^1$ or $a_{\neg x}^2$ that features in the round. 
Agent $a_{\neg x}^1$ features if $c$ is the first clause in which literal $x$ is present. Agent $a_{\neg x}^2$ features if $c$ is the second clause in which literal $x$ is present. Recall that each literal occurs in exactly two clauses in the formula.         
         After all the clause rounds are finished, agent $1$ gets $|C|$ turns to get a chance to get clause items in case they are available. 
         
         For agent $1$, the {relevant} preferences in the clause round are:
 		\begin{align*}
 			1:&\quad o_{c}^1
 			\end{align*}
	
		For a variable $x$, if $x$ is a literal in the clause, the relevant preferences in this round are:
		
		\begin{align*}
			a_{\neg x}^1:&\quad h_{x}^2, h_{x}^3, o_{c}^3, o_c^2, o_c^1 \text{ such that } x\in c_j \\
			a_{\neg x}^2:&\quad  h_{x}^2, h_{x}^3, o_{c'}^3, o_{c'}^2, o_{c'}^1 \text{ such that } x\in c'\neq c\\
			\end{align*}
            
            		For a variable $x$, if $\neg x$ is a literal in the clause, the relevant preferences in this round are:

		\begin{align*}
			a_{x}^1:&\quad d_{\neg x}^{12}, o_{c}^3, o_c^2, o_c^1 \text{ such that } \neg x\in c_j \\
			a_{x}^2:&\quad  d_{\neg x}^{22}, o_{c'}^3, o_{c'}^2, o_{c'}^1 \text{ such that } \neg x\in c'\neq c\\
			\end{align*}

Please note that the items $h_{x}^2, h_{x}^3$ are the consistency items that featured in relevant items in the choice round corresponding to variable $x$ and $d_{\neg x}^{12}$ and $d_{\neg x}^{22}$ are the dummy items that featured in relevant items in the choice round corresponding to variable $x$. %If agent $1$ has made a consistent choice in the choice round corresponding to variable $x$, then either $h_{x}^2, h_{x}^3$ feature in the clause round or $h_{x}^2, h_{x}^3$ feature in the clause round. 
If agent $1$ has made a consistent choice in the choice round corresponding to variable $x$ and set $x$ to ``true'', then $a_{x_i}^1$ and $a_{x_i}^2$ already want to get clause items in their respective clause rounds. If agent $1$ has made a consistent choice in the choice round corresponding to variable $x$ and set $x$ to ``false'', then $a_{\neg x_i}^1$ and $a_{\neg x_i}^2$ already want to get clause items in their respective clause rounds. 

In the clause round, any literal that is not satisfied, the agent corresponding to literal gets a clause item. So for example if literal $x$ is false, then $a_{\neg x}^1$ gets a clause item for the first clause in which $x$ is present.   
$a_{x}^2$ gets a clause item for the second clause in which $x$ is present. Therefore, if all the literals of a clause $c$ are false, then agent $1$ does not get $o_c^1$.

If literal $x$ is satisfied, the agent $a_{\neg x}$ corresponding to it gets a dummy or consistency item in that round instead of a clause item.  This means that if all the literals of a clause are not satisfied, then all three clause items of a clause are gone, and agent $1$ does not get a clause item. He is only {interested} in one of the clause items $o_c^1$. The other clause items are far down in his preference list so he would rather get all the top clause item $o_c^1$ for each clause $c$ rather than $o_c^2$ and $o_c^3$.

Let us assume that for clause $c=(x_{i}\vee \neg x_{j}\vee \neg x_{k})$, variables $x_{i}$, $ x_{j}$ and  $ x_{k}$ are set to  true
 i.e., agent $1$ got choice items corresponding to the $\neg x_i$, $\neg x_j$ and $\neg x_k$. This means that in the clause round, $a_{ x_{j}}^1$ and  $a_{x_{k}}^1$ are ready to take their clause items $o_c^3$ and $o_c^2$ but $a_{\neg x_i}^1$ wants to get one of the unallocated consistency items before he is interested in consistency items $o_c^3, o_c^2,o_c^1$. This is helpful for agent $1$ because he can get $o_c^1$.
Since each literal occurs exactly twice in the formula, note that as long as $1$ makes a consistent choice, there will be another clause $c'$ in which literal $x$ is present and if $x$ is set to true, then $a_{\neg x}^2$ will get $h_{x_i}^3$ and hence $1$ will be able to get $o_{c'}$. Note that after the clause rounds all the consistency {items} are already consumed so the agent $1$ can hope to get all the top clause items if they were not already taken in the clause rounds.

\begin{table}[h!]
	\centering
		\scalebox{0.63}{
\begin{tabular}{lccc|c}%{| c | c | c | c | c | c | c | c | c |c | c | c |}
  \hline
\textbf{Stage}&  1 & 2 &3 &\text{after clause rounds}\\
  \hline 
\textbf{Agent }&$a_{\neg x_i}^1$&$a_{x_j}^1$&$a_{x_k}^1$&1\\
 \hline
\textbf{Item picked}&$h_{x_i}^2$ &$o_{c}^3$&$o_{c}^2$&$o_{c}^1$\\
  \hline
\end{tabular}
}
\caption{Clause round corresponding to clause $c=(x_{i}\vee \neg x_{j}\vee \neg x_{k})$}
\end{table}

\paragraph{Collection round}

		The sequence in the collection round 

		 \[\underbrace{1,\ldots, 1}_{|C|}\]%,\underbrace{2,\ldots, 2}_{|C|}\] 
         
         The relevant preferences are:
 		\begin{align*}
 			1:&\quad {o_{c_{1}}},\ldots, {o_{c_{|C|}}}
            %2:&\quad {o_{c_{1}}},\ldots, {o_{c_{|C|}}}%,l_1,\ldots, l_{|C|}
 			\end{align*}
            
The idea is that if agent $1$ make choices that sets all the clauses as true, then agent $1$ gets all the clause items. Note that if $1$ makes a consistent choice for the variables but does not pick up all the clause items in the collection round (because the formula is unsatisfiable), then $1$ does not get all the clause items. Since there are items less preferred by $1$ than the $o_c^1$s such as $o_c^2$ and $o_c^3$s, agent $1$ is forced to pick a much less preferred item in the collection round.

\paragraph{Utility of Agent $1$}

The utility function $u_1$ of agent $1$ is specified as follows:

\begin{itemize}
	\item There is negligible difference between $u_1(o_{x_i}^1)$ and $u_1(o_{\neg x_i}^1)$. The utility of both is considerably more than $u_1(o_{x_i}^2)$ and $u_1(o_{\neg x_i}^2)$. There is negligible difference between $u_1(o_{x_i}^2)$ and $u_1(o_{\neg x_i}^2)$.
	\item 
For any variable $x_i$, agent $1$'s preferences over the consistency items are as follows:
\begin{align*}
1:&\quad  h_{\neg x_i}^1,h_{\neg x_i}^2,h_{\neg x_i}^3,h_{x_i}^1,h_{x_i}^2,h_{x_i}^3  
\end{align*}
 The utility is set as follows:
 $u_1(h_{x_i}^2)+u_1(h_{\neg x_i}^2)<u_1(h_{\neg x_i}^1)+u_1(h_{ x_i}^3)={u_1(h_{x_i}^1)+u_1(h_{\neg x_i}^3)}$. 
	\item All items that agent $1$ is considering getting in a  round (choice or clause) are considerably more preferred than the corresponding items in the latter rounds. 
	\item The utilities are set in a way so that as long as agent $1$ gets two items corresponding to a variable, at least one top choice consistency items in each round and his target clause item in each clause round, agent $1$ gets utility at least $T$. 
\end{itemize}

Based on construction of the choice and clause rounds, we are in a position to prove a series of claims. 
\begin{claim}
If agent $1$ does not make a consistent choice of the variable items, then it does not get utility $T$.
\end{claim}
\begin{proof}
	If $1$ does not make a choice in each choice round, his most preferred items corresponding to the literals are taken by the agents corresponding to the literal. If $1$ makes a choice in each choice round but does not make a consistent choice, then he gets $\{h_{\neg x}^2,h_{x}^2\}$ which has much less utility than $\{h_{\neg x}^1,h_{x}^3\}$ or $\{h_{x}^1,h_{\neg x}^3\}$ 
which means he cannot get total utility $T$.
	\end{proof}
	
	\begin{claim}
	If agent $1$ makes consistent choices but the assignment is not satisfying, then agent $1$ does not get utility $T$.
	\end{claim}
	\begin{proof}
		If some clause $c$ is set false, then agent $1$ is not able to $o_c^1$ because the literal agents in the clause round corresponding to $c$ take all the items $o_c^3,o_c^2,o_c^1$. This means that agent $1$ does not get utility $T$. 
		\end{proof}

		\begin{claim}
			If there exists a satisfying assignment, then agent $1$ can get utility $T$.
			\end{claim}
	\begin{proof}
		If there exists a satisfying assignment, then consider the preference report of agent $1$ in which in each choice round, he picks $o_{x_i}^1$ and $o_{x_i}^2$ if the $x_i$ is set to be false. By doing this he gets to pick a top consistency item in that round as well. Since all the clauses are satisfied, in each clause round, agent $1$ is able to gets his clause item $o_{c}^1$.
		The utilities are set in a way so that as long as agent $1$ gets two items corresponding to the same literal and hence at least one top choice consistency items in each round and his target clause item in each clause round, agent $1$ gets utility at least $T$. 
		\end{proof}

The claims show that agent $1$ gets utility at least $T$ if and only if there is a satisfying truth assignment. 
		\end{proof}

    % \begin{remark}
    %     The proof above shows that computing a best response is NP-complete. The complexity of the following problems is still unresolved: (1) check whether there exists a better response to the truthful profile (2) check whether there exists a better response to the truthful profile that is better for all utilities consistent with the ordinal preferences (gives an allocation that is more preferred with respect to responsive preferences).
    %     \end{remark}

 \begin{example}    \label{example:sat}
     We illustrate the reduction in the proof of  Theorem~\ref{th:main}. 
 For the following SAT formula, we illustrate how we build an allocation setting with the agent set, item set, preferences of agents and the picking sequence.
         % \vspace{-1em}
         \begin{equation*}\small \underbrace{(x_1\vee x_2\vee x_3)}_{c_1}
        \underbrace{(\neg x_1\vee \neg x_2\vee \neg x_3)}_{c_2}
            \underbrace{(x_1\vee \neg x_2\vee x_3)}_{c_3}
            \underbrace{(\neg x_1\vee x_2\vee \neg x_3)}_{c_4}
        \end{equation*}

Set of agents is $N=\{1\}\cup \{ a_{x_i}^1, a_{x_i}^2,
a_{\neg x_i}^1, a_{\neg x_i}^2\midd i\in \{1,2,3\} \}$.

Set of items is
$O=\{o_{x_i}^1, o_{x_i}^2,o_{\neg x_i}^1, o_{\neg x_i}^2, h_{x_i}^1, h_{x_i}^2,h_{\neg x_i}^1, h_{\neg x_i}^2, {d_{x_i}^{11}}, {d_{x_i}^{12}}, {d_{x_i}^{21}}, {d_{x_i}^{22}}\midd i\in \{1,2,3\}\}$

		\begin{align*}
1:&\quad o_{x_1}^1, o_{\neg x_1}^1, o_{x_1}^2, o_{\neg x_1}^2,  h_{\neg x_1}^1,h_{\neg x_1}^2,h_{\neg x_1}^3,h_{x_1}^1,h_{x_1}^2,h_{x_1}^3,\\
&\quad o_{x_2}^1, o_{\neg x_2}^1, o_{x_2}^2, o_{\neg x_2}^2,  h_{\neg x_2}^1,h_{\neg x_2}^2,h_{\neg x_2}^3,h_{x_2}^1,h_{x_2}^2,h_{x_2}^3\\
&\quad o_{x_3}^1, o_{\neg x_3}^1, o_{x_3}^2, o_{\neg x_3}^2,  h_{\neg x_3}^1,h_{\neg x_3}^2,h_{\neg x_3}^3,h_{x_3}^1,h_{x_3}^2,h_{x_3}^3\\
                          &\quad o_{c_1}^1,o_{c_2}^1,o_{c_3}^1,o_{c_4}^1\\
             \midrule
 			a_{\neg x_1}^1:&\quad o_{x_1}^1, d_{x_1}^{11}, d_{x_1}^{12},o_{x_1}^2, h_{x_1}^1,  h_{x_1}^2,h_{x_1}^3,o_{c_1}^3, o_{c_1}^2, o_{c_1}^1\\
 			a_{\neg x_1}^2:&\quad d_{x_1}^{21}, o_{x_1}^1, o_{x_1}^2,{d_{x_i}^{22}}, h_{x_1}^1,  h_{x_1}^2,h_{x_1}^3,o_{c_3}^3, o_{c_3}^2, o_{c_3}^1\\
 			a_{\neg x_2}^1:&\quad o_{x_2}^1, d_{x_2}^{11}, d_{x_2}^{12},o_{x_2}^2, h_{x_2}^1,  h_{x_2}^2,h_{x_2}^3,o_{c_1}^3, o_{c_1}^2, o_{c_1}^1\\
 			a_{\neg x_2}^2:&\quad d_{x_2}^{21}, o_{x_2}^1, o_{x_2}^2,{d_{x_i}^{22}}, h_{x_2}^1,  h_{x_2}^2,h_{x_2}^3,o_{c_4}^3, o_{c_4}^2, o_{c_4}^1\\
 			a_{\neg x_3}^1:&\quad o_{x_3}^1, d_{x_3}^{11}, d_{x_3}^{12},o_{x_3}^2, h_{x_3}^1,  h_{x_3}^2,h_{x_3}^3,o_{c_1}^3, o_{c_1}^2, o_{c_1}^1\\
 			a_{\neg x_3}^2:&\quad d_{x_3}^{21}, o_{x_3}^1, o_{x_3}^2,{d_{x_i}^{22}}, h_{x_3}^1,  h_{x_3}^2,h_{x_3}^3,o_{c_3}^3, o_{c_3}^2, o_{c_3}^1\\
            \midrule
 			a_{x_1}^1:&\quad o_{\neg x_1}^1, d_{\neg x_1}^{11}, h_{\neg x_1}^1,o_{\neg x_1}^2, h_{\neg x_1}^2,h_{\neg x_1}^3,d_{\neg x_1}^{12}o_{c_2}^3, o_{c_2}^2, o_{c_2}^1\\
 			a_{x_1}^2:&\quad d_{\neg x_1}^{21}, o_{\neg x_1}^1, o_{\neg x_1}^2, h_{\neg x_1}^1,h_{\neg x_1}^2,h_{\neg x_1}^3,d_{\neg x_1}^{22},o_{c_4}^3, o_{c_4}^2, o_{c_4}^1\\
 			a_{x_2}^1:&\quad o_{\neg x_2}^1, d_{\neg x_2}^{11}, h_{\neg x_2}^1,o_{\neg x_2}^2, h_{\neg x_2}^2,h_{\neg x_2}^3,d_{\neg x_2}^{12},o_{c_2}^3, o_{c_2}^2, o_{c_2}^1\\
 			a_{x_2}^2:&\quad d_{\neg x_2}^{21}, o_{\neg x_2}^1, o_{\neg x_2}^2, h_{\neg x_2}^1,h_{\neg x_2}^2,h_{\neg x_2}^3,d_{\neg x_2}^{22},o_{c_3}^3, o_{c_3}^2, o_{c_3}^1\\
 			a_{x_3}^1:&\quad o_{\neg x_3}^1, d_{\neg x_3}^{11}, h_{\neg x_3}^1,o_{\neg x_3}^2, h_{\neg x_3}^2,h_{\neg x_3}^3,d_{\neg x_3}^{12},o_{c_2}^3, o_{c_2}^2, o_{c_2}^1\\
 			a_{x_3}^2:&\quad d_{\neg x_3}^{21}, o_{\neg x_3}^1, o_{\neg x_3}^2, h_{\neg x_3}^1,h_{\neg x_3}^2,h_{\neg x_3}^3,d_{\neg x_3}^{22},o_{c_4}^3, o_{c_4}^2, o_{c_4}^1\\
                                          \midrule                      
			\end{align*}

The picking sequence is as follows

		\begin{align*}
\text{choice round 1}\quad &1, a_{\neg x_1}^1, a_{\neg x_1}^2, a_{x_1}^1, a_{ x_1}^2, 1, a_{\neg x_1}^1, a_{\neg x_1}^2, a_{\neg x_1}^1, a_{\neg x_1}^2, a_{\neg x_1}^1, a_{\neg x_1}^2,1, a_{x_1}^1, a_{x_1}^2, 1\\
\text{choice round 2}\quad &1, a_{\neg x_2}^1, a_{\neg x_2}^2, a_{x_2}^1, a_{ x_2}^2, 1, a_{\neg x_2}^1, a_{\neg x_2}^2, a_{\neg x_2}^1, a_{\neg x_2}^2, a_{\neg x_2}^1, a_{\neg x_2}^2,1, a_{x_2}^1, a_{x_2}^2, 1\\
\text{choice round 3}\quad &1, a_{\neg x_3}^1, a_{\neg x_3}^2, a_{x_3}^1, a_{ x_3}^2, 1, a_{\neg x_3}^1, a_{\neg x_3}^2, a_{\neg x_3}^1, a_{\neg x_3}^2, a_{\neg x_3}^1, a_{\neg x_3}^2,1, a_{x_3}^1, a_{x_3}^2, 1\\
\text{clause round 1}\quad &a_{\neg x_1}^1, a_{\neg  x_2}^1,a_{\neg  x_3}^1\\
\text{choice round 2}\quad &a_{x_1}^1, a_{x_2}^1,a_{x_3}^1\\
\text{choice round 3}\quad &a_{\neg  x_1}^2, a_{x_2}^2,a_{\neg  x_3}^2\\
\text{choice round 4}\quad &a_{x_1}^2, a_{\neg  x_2}^2,a_{x_3}^2\\
\text{collection round}\quad &1,1,1,1%,2,2,2,2
	\end{align*}

The formula is satisfiable if $x_1$ is true, $x_2$ is false and $x_3$ is false. 
Let us show how the allocation looks like when agent $1$ picks items according to the truth assignment.

	\begin{table}[h!]
		\centering
			\scalebox{0.63}{
	\begin{tabular}{l|ccccc|ccccc|cccccc}%{| c | c | c | c | c | c | c | c | c |c | c | c |}
	  \hline
	\textbf{Stage}&  1 & 2 &3 &4&5&6&7&8&9&10&11&12&13&14&15&16\\
	  \hline 
	\textbf{Agent }&1& $a_{\neg x_1}^1$&$a_{\neg x_1}^2$ &$a_{x_1}^1$&$a_{ x_1}^2$&1&$a_{\neg x_1}^1$&$a_{\neg x_1}^2$&$a_{x_1}^1$&$a_{x_1}^2$&$a_{\neg x_1}^1$&$a_{\neg x_1}^2$&1&$a_{x_1}^1$&$a_{x_1}^2$&1  \\
	 \hline
	\textbf{Item picked}& $o_{\neg x_1}^1$& $o_{x_1}^1$&$d_{x_1}^{21}$ &$d_{\neg x_1}^{11}$&$d_{\neg x_1}^{21}$&$o_{\neg x_1}^2$&$d_{x_1}^{11}$&$o_{x_1}^2$&$h_{\neg x_1}^{1}$&$h_{\neg x_1}^{2}$&$d_{x_1}^{12}$&$d_{x_1}^{22}$&$h_{\neg x_1}^{3}$&$d_{\neg x_1}^{12}$&$d_{\neg x_1}^{22}$&$h_{x_1}^1$  \\
	  \hline
	\end{tabular}
	}
	\caption{Choice round 1 for variable $x_1$ in which agent $1$ makes consistent choice $o_{\neg x_1}^1$ and $o_{\neg x_1}^2$ so that variable $x_1$ is set to true.}
	\label{table:choiceround1}
	\end{table}

	\begin{table}[h!]
		\centering
			\scalebox{0.63}{
	\begin{tabular}{l|ccccc|ccccc|cccccc}%{| c | c | c | c | c | c | c | c | c |c | c | c |}
	  \hline
	\textbf{Stage}&  1 & 2 &3 &4&5&6&7&8&9&10&11&12&13&14&15&16\\
	  \hline 
	\textbf{Agent }&1& $a_{\neg x_2}^1$&$a_{\neg x_2}^2$ &$a_{x_2}^1$&$a_{ x_2}^2$&1&$a_{\neg x_2}^1$&$a_{\neg x_2}^2$&$a_{x_2}^1$&$a_{x_2}^2$&$a_{\neg x_2}^1$&$a_{\neg x_2}^2$&1&$a_{x_2}^1$&$a_{x_2}^2$&1  \\
	 \hline
	\textbf{Item picked}& $o_{x_2}^1$& $d_{x_2}^{11}$&$d_{x_2}^{21}$ &$o_{\neg x_2}^1$&$d_{\neg  x_2}^{21}$&$o_{x_2}^2$&$d_{x_2}^{12}$&$d_{x_2}^{22}$&$d_{\neg  x_2}^{11}$&$o_{\neg x_2}^2$&$h_{x_2}^1$&$h_{x_2}^2$&$h_{\neg x_2}^1$&$h_{\neg  x_2}^{2}$&$h_{\neg  x_2}^{3}$&$h_{x_2}^3$  \\
	  \hline
	\end{tabular}
	}
	\caption{Choice round 2 for variable $x_2$ in which agent $1$ makes consistent choice $o_{x_2}^1$ and $o_{x_2}^2$ so that variable $x_2$ is set to false. }
	\label{table:choiceround2}
	\end{table}

	\begin{table}[h!]
		\centering
			\scalebox{0.63}{
	\begin{tabular}{l|ccccc|ccccc|cccccc}%{| c | c | c | c | c | c | c | c | c |c | c | c |}
	  \hline
	\textbf{Stage}&  1 & 2 &3 &4&5&6&7&8&9&10&11&12&13&14&15&16\\
	  \hline 
	\textbf{Agent }&1& $a_{\neg x_3}^1$&$a_{\neg x_3}^2$ &$a_{x_3}^1$&$a_{ x_3}^2$&1&$a_{\neg x_3}^1$&$a_{\neg x_3}^2$&$a_{x_3}^1$&$a_{x_3}^2$&$a_{\neg x_3}^1$&$a_{\neg x_3}^2$&1&$a_{x_3}^1$&$a_{x_3}^2$&1  \\
	 \hline
	\textbf{Item picked}& $o_{x_3}^1$& $d_{x_3}^{11}$&$d_{x_3}^{21}$ &$o_{\neg x_3}^1$&$d_{\neg  x_3}^{21}$&$o_{x_3}^2$&$d_{x_3}^{12}$&$d_{x_3}^{22}$&$d_{\neg  x_3}^{11}$&$o_{\neg x_3}^2$&$h_{x_3}^1$&$h_{x_3}^2$&$h_{\neg x_3}^1$&$h_{\neg  x_3}^{2}$&$h_{\neg  x_3}^{3}$&$h_{x_3}^3$  \\
	  \hline
	\end{tabular}
	}
	\caption{Choice round 3 for variable $x_3$ in which agent $1$ makes consistent choice $o_{x_3}^1$ and $o_{x_3}^2$ so that variable $x_3$ is set to false. }
	\label{table:choiceround3}
	\end{table}
    
    \begin{table}[h!]
    	\centering
    		\scalebox{0.63}{
    \begin{tabular}{lcccc}%{| c | c | c | c | c | c | c | c | c |c | c | c |}
      \hline
    \textbf{Stage}&  1 & 2 &3\\
      \hline 
    \textbf{Agent }&$a_{\neg x_1}^1$&$a_{\neg x_2}^1$&$a_{\neg x_3}^1$\\
     \hline
    \textbf{Item picked}&$h_{x_1}^2$ &$o_{c_1}^3$&$o_{c_1}^3$\\
      \hline
    \end{tabular}
    }
    \caption{Clause round 1}
    \end{table}
    
    \begin{table}[h!]
    	\centering
    		\scalebox{0.63}{
    \begin{tabular}{lcccc}%{| c | c | c | c | c | c | c | c | c |c | c | c |}
      \hline
    \textbf{Stage}&  1 & 2 &3\\
      \hline 
    \textbf{Agent }&$a_{x_1}^1$&$a_{x_2}^1$&$a_{x_3}^1$\\
     \hline
    \textbf{Item picked}&$o_{c_2}^3$ &$d_{\neg x_2}^{12}$&$d_{\neg x_3}^{12}$\\
      \hline
    \end{tabular}
    }
    \caption{Clause round 2}
    \end{table}
    
    \begin{table}[h!]
    	\centering
    		\scalebox{0.63}{
    \begin{tabular}{lcccc}%{| c | c | c | c | c | c | c | c | c |c | c | c |}
      \hline
    \textbf{Stage}&  1 & 2 &3\\
      \hline 
    \textbf{Agent }&$a_{\neg x_1}^2$&$a_{x_2}^2$&$a_{\neg x_3}^2$\\
     \hline
    \textbf{Item picked}&$h_{x_1}^3$ &$d_{\neg x_2}^{22}$&$o_{c_3}^2$\\
      \hline
    \end{tabular}
    }
    \caption{Clause round 3}
    \end{table}
    
    \begin{table}[h!]
    	\centering
    		\scalebox{0.63}{
    \begin{tabular}{lcccc}%{| c | c | c | c | c | c | c | c | c |c | c | c |}
      \hline
    \textbf{Stage}&  1 & 2 &3\\
      \hline 
    \textbf{Agent }&$a_{x_1}^2$&$a_{\neg x_2}^2$&$a_{x_3}^2$\\
     \hline
    \textbf{Item picked}&$o_{c_4}^2$ &$o_{c_4}^2$&$d_{\neg x_3}^{22}$\\
      \hline
    \end{tabular}
    }
        \caption{Clause round 4}
    % \caption{Clause round corresponding to clause $c=(x_{i}\vee \neg x_{j}\vee \neg x_{k})$}
    \end{table}

    \begin{table}[h!]
    	\centering
    		\scalebox{0.63}{
    \begin{tabular}{lcccccccc}%{| c | c | c | c | c | c | c | c | c |c | c | c |}
      \hline
    \textbf{Stage}&  1 & 2 &3&4\\
      \hline 
    \textbf{Agent }&$1$&$1$&$1$&$1$\\
     \hline
    \textbf{Item picked}&$o_{c_1}^1$ &$o_{c_2}^1$&$o_{c_3}^1$&$o_{c_4}^1$\\
      \hline
    \end{tabular}
    }
        \caption{Collection round}
    % \caption{Clause round corresponding to clause $c=(x_{i}\vee \neg x_{j}\vee \neg x_{k})$}
    \end{table}

         \end{example}

%%%end of long version
\section{Computing a Best Response: Case of Two Agents}

The insights and results of \citet{BoLa14b} in their Section 3.1 and 3.2 still hold for $n=2$. 
For the sake of completeness, we outline the main arguments. 
We can ascertain that for  $n=2$, not only is there is a linear-time algorithm to compute a best response but each best response results in the same allocation irrespective of the cardinal utilities consistent with the ordinal preferences.

% We first prove or reprove a series of lemmas.

We present a series of lemmas.

			\begin{lemma}[From Proposition 7 by \citet{BoLa11a}]
				\label{lemma:simple-strategy}				Let the set of items $S=\{a_1,\ldots, a_{m_1}\}$ be such that $a_1\succ_2 a_2 \cdots \succ_2 a_{m_1}$. If $S$ is achievable by $1$ for some preference report, then agent $1$ can achieve $S$ by reporting $a_1\succ_1  a_2 \succ_1  \cdots a_{m_1}\succ_1 \{\text{all other items}\}$.
				\end{lemma}
				\begin{proof}
				Consider any report $\succ_1'$ such that the result of $(\succ_1',\succ_2$ is that agent $1$ gets $S$. Then changing $\succ_1'$ to $\succ_1''$ in which items in $O\setminus S$ are moved to items after $S$ does not change the outcome. At each stage when agent $i\in \{1,2\}$ picks, the most preferred available item of $i$ remains the same.

We now show that if agent $1$ can achieve $S$, he can do so by
reporting $a_1\succ_1  a_2 \cdots a_{m_1}\succ_1 \{\text{all other items}\}$.	Assume that agent $1$ does not get $S$ by report  $a_1\succ_1  a_2 \cdots a_{m_1}\cdots \succ_1 \{\text{all other items}\}$. We show that agent $1$ cannot get $S$ by any other report.
Let us consider the earliest stage in which agent $2$ gets an item from $S$. Let the item be $a_i$ in stage $\ell$. This means that by stage $\ell$ agent $2$ did not any items from $\{a_1,\ldots, a_{i-1}\}$ which agent $1$ got.
Note that by stage $\ell$, agent $1$ gets $i-1$ picks and agent $2$ gets $\ell-(i-1)=\ell=i+1$ picks.
Note that $\{a\midd a\succ_2 a_{i}\}$ is the union of $\{a_1,\ldots, a_{i-1}\}$ and the items in $O\setminus S$ that agent $2$ got before stage $\ell$.
The most preferred $\ell$ items of agent $2$ include $i$ items from $S$ and $\ell-1-(i-1)=\ell-i$ other items. Since agent $2$ has $\ell-i+1$ picks, he will be able to get one item from $\{a_1,\ldots, a_{i}\}$ irrespective of what agent $1$ reports because by stage $\ell$, agent $1$ has only $i-1$ picks.
							\end{proof}

			\begin{lemma}[From Proposition 7 by \citet{BoLa11a}]
				Let the set of items $S=\{a_1,\ldots, a_{m_1}\}$ be such that $a_1\succ_2 a_2 \cdots \succ_2 a_{m_1}$. Then the following conditions are equivalent:
\begin{enumerate}
	\item $S$ is achievable. 
	\item agent $1$ can achieve $S$ by reporting $a_1\succ_1  a_2 \succ_1  \cdots a_{m_1}\cdots \succ_1 \{\text{all other items}\}$.
	\item when agent $1$ reports $a_1\succ_1  a_2 \succ_1  \cdots a_{m_1}\cdots \succ_1 \{\text{all other items}\}$, for each picking stage $\ell$ in which agent $1$ picks his $i$-th, all the $\ell-i$ items allocated to $2$ by stage $\ell$ more preferred than $a_i$.
\end{enumerate}
							\end{lemma}
				\begin{proof}
					$(ii)$ trivially implies $(i)$. Lemma~\ref{lemma:simple-strategy} shows that $(ii)$ implies $(i)$.

					We prove that $(iii)$ implies $(ii)$
			Let us assume that for each picking stage $\ell$ in which agent $1$ picks his $i$-th, all the $\ell-i$ items allocated to $2$ by stage $\ell$ more preferred than $a_i$. Then agent $2$ is always busy getting more preferred items and agent $1$ get $a_i$ in his $i$-th pick.

					We prove that $(ii)$ implies $(iii)$.
				Now assume that in some picking stage $\ell$ in which agent $1$ picks his $i$-th item ($a_i$), not all the $\ell-i$ items allocated to $2$ by stage $\ell$ are more preferred than $a_i$. But this means that $2$ would have picked $a_i$ already by stage $\ell$.
					\end{proof}

					\begin{lemma}[Proposition 9 of \citet{BoLa11a}]\label{lemma:br-lex}
						For two agents, lexicographic best response is polynomial-time computable.
						\end{lemma}
					\begin{proof}
						Let us assume that the preferences of the manipulator are $o_1,o_2,\ldots, o_m$.
						We set the target set $S$ to empty and $O'$ to $O$. 
						Take the most preferred item  $o\in O'$ and check whether $S\cup \{o\}$ is achievable by 1. If yes, we append $o$ to $S$. In either case we delete $o$ from $O'$. We continue in this fashion until $O'$ is empty. 
					For $n=2$, it can be easily checked whether a given subset $S$ of items is achievable by letting agent 1 express the items in $S$ in the same order of preferences as agent 2's preferences over $S$. This follows from Lemma~\ref{lemma:simple-strategy}.
						\end{proof}

					% \begin{lemma}
% 						Let the set of items $S=\{a_1,\ldots, a_{m_1}\}$ be such that $a_1\succ_2 a_2 \cdots \succ_2 a_{m_1}$.
% 			Then agent $1$ can achieve $S$ by reporting $a_1\succ_1  a_2 \cdots a_{m_1}\cdots \succ_1 \{\text{all other items}\}$ if and only if
% 			for each picking stage $\ell$ in which agent $1$ picks his $i$-th, all the $\ell-i$ items allocated to $2$ by stage $\ell$ more preferred than $a_i$.
% 							\end{lemma}
% 							\begin{proof}
% 						Let us assume that for each picking stage $\ell$ in which agent $1$ picks his $i$-th, all the $\ell-i$ items allocated to $2$ by stage $\ell$ more preferred than $a_i$. Then agent $2$ is always busy getting more preferred items and agent $1$ get $a_i$ in his $i$-th pick.
%
% 							Now assume that in some picking stage $\ell$ in which agent $1$ picks his $i$-th item ($a_i$), not all the $\ell-i$ items allocated to $2$ by stage $\ell$ are more preferred than $a_i$. But this means that $2$ would have picked $a_i$ already by stage $\ell$.
% 								\end{proof}

	Next, we see that for $n=2$, the outcome of any best response is the same for the manipulator. %We can prove this by showing that if there are two different allocations of the agent that are outcomes of best responses, then neither of them is an outcome of a best response. 
    The argument of \citet{BoLa14b}  works  as it is directly for the case of two agents. 

	\begin{lemma}[Lemma 1 of \citet{BoLa14b}]\label{lemma:achieve}
		Let $A$ and $B$ be achievable by agent 1. Let $a=\max_{\pref_1}((A\setminus B)\cup (B\setminus A))$ and assume that $a\in A$. Let $b=\max_{\pref_1}(B\setminus A)$. Then $B\cup \{a\}\setminus \{b\}$ is achievable for agent $1$.
		\end{lemma}

            Based on the lemmas above, the following theorem can be proved.

			\begin{theorem}[\citet{BoLa14a,BoLa14b}]
				For $n=2$, there exists a polynomial-time algorithm to compute a best response. Furthermore the allocation of responding agent as a result of the best response is unique. 
				\end{theorem}
				\begin{proof}
					From Lemma~\ref{lemma:achieve}, we know that in order to compute a best response for a given utility function consistent with the ordinal preferences, it is sufficient to compute the best response for \emph{any} utility function consistent with the ordinal preferences. We know from Lemma~\ref{lemma:br-lex} that there exists a polynomial-time algorithm to compute the best response for lexicographic utilities. 
					\end{proof}
					
					\begin{corollary}
						For $n=2$, there exists a polynomial-time algorithm to verify a pure Nash equilibrium. 
						\end{corollary}
\begin{proof}
	If there exists a polynomial-time algorithm to compute a best response, it can be used to compute the best response of each agent. A profile is in pure Nash equilibrium if and only if the best response of each agent yields at most the same utility as the preference reported in the given preference profile.  
	\end{proof}

    \section{Conclusions}
        
        In this paper, we showed that computing a best response under sequential allocation to maximize additive utility is NP-hard. The result is surprising because previously it has been claimed in the literature (COMSOC 2014 and ECAI 2014) that the problem admits a polynomial-time algorithm. Our NP-hardness result does not involve a constant number of agents. It remains an interesting open problem whether manipulating sequential allocation is NP-hard when the number of agents is three or some other constant. 
        
        \subsubsection*{Acknowledgment}
     Haris Aziz is funded by the Australian Government through the Department of Communications and the Australian Research Council through the ICT Centre of Excellence Program.

\bibliographystyle{model1a-num-names}
  % \bibliography{../../pamas/abbshort,../../pamas/pamas,../../pamas/brandt,../../pamas/aziz}

\begin{thebibliography}{14}
        \expandafter\ifx\csname natexlab\endcsname\relax\def\natexlab#1{#1}\fi
        \expandafter\ifx\csname url\endcsname\relax
          \def\url#1{\texttt{#1}}\fi
        \expandafter\ifx\csname urlprefix\endcsname\relax\def\urlprefix{URL }\fi

        \bibitem[{Aziz et~al.(2015{\natexlab{a}})Aziz, Gaspers, Mackenzie, Mattei,
          Narodytska, and Walsh}]{AGM+15c}
        Aziz, H., Gaspers, S., Mackenzie, S., Mattei, N., Narodytska, N., Walsh, T.,
          2015{\natexlab{a}}. Manipulating the probabilistic serial rule. In: Proc.~of
          14th AAMAS Conference. pp. 1451--1459.

        \bibitem[{Aziz et~al.(2015{\natexlab{b}})Aziz, Walsh, and Xia}]{AWX15b}
        Aziz, H., Walsh, T., Xia, L., 2015{\natexlab{b}}. Possible and necessary
          allocations via sequential mechanisms. In: Proc.~of 23rd IJCAI. pp. 468--474.

        \bibitem[{Bouveret and Lang(2011)}]{BoLa11a}
        Bouveret, S., Lang, J., 2011. A general elicitation-free protocol for
          allocating indivisible goods. In: Proc.~of 22nd IJCAI. AAAI Press, pp.
          73--78.

        \bibitem[{Bouveret and Lang(2014{\natexlab{a}})}]{BoLa14a}
        Bouveret, S., Lang, J., 2014{\natexlab{a}}. Manipulating picking sequences. In:
          Proc.~of 5th COMSOC.

        \bibitem[{Bouveret and Lang(2014{\natexlab{b}})}]{BoLa14b}
        Bouveret, S., Lang, J., 2014{\natexlab{b}}. Manipulating picking sequences. In:
          Proc.~of 21st ECAI. pp. 141--146.

        \bibitem[{Brams and Straffin(1979)}]{BrSt79a}
        Brams, S.~J., Straffin, P.~D., 1979. Prisoners' dilemma and professional sports
          drafts. The American Mathematical Monthly 86~(2), 80--88.

        \bibitem[{Brams and Taylor(1996)}]{BrTa96a}
        Brams, S.~J., Taylor, A.~D., 1996. Fair Division: From Cake-Cutting to Dispute
          Resolution. Cambridge University Press.

        \bibitem[{Brill and Conitzer(2016)}]{Wals16a}
        Brill, M., Conitzer, V., 2016. Strategic behaviour when allocating indivisible
          goods. In: Proc.~of 30th AAAI Conference. AAAI Press.

        \bibitem[{Bulteau et~al.(2015)Bulteau, Chen, Faliszewski, Niedermeier, and
          Talmon}]{BCF+15a}
        Bulteau, L., Chen, J., Faliszewski, P., Niedermeier, R., Talmon, N., 2015.
          Combinatorial voter control in elections. Theoretical Computer Science 589,
          99--120.

        \bibitem[{Hosseini and Larson(2015)}]{HoLa15a}
        Hosseini, H., Larson, K., 2015. Strategyproof quota mechanisms for multiple
          assignment problems. Tech. Rep. 1507.07064, arXiv.org.

        \bibitem[{Kalinowski et~al.(2013)Kalinowski, Narodytska, and Walsh}]{KNW13a}
        Kalinowski, T., Narodytska, N., Walsh, T., 2013. A social welfare optimal
          sequential allocation procedure. In: Proc.~of 22nd IJCAI. AAAI Press, pp.
          227--233.

        \bibitem[{Kohler and Chandrasekaran(1971)}]{KoCh71a}
        Kohler, D.~A., Chandrasekaran, R., 1971. A class of sequential games.
          Operations Research 19~(2), 270--277.

        \bibitem[{Lang and Rothe(2015)}]{LaRe15a}
        Lang, J., Rothe, J., 2015. Fair division of indivisible goods. In: Economics
          and Computation An Introduction to Algorithmic Game Theory, Computational
          Social Choice, and Fair Division. Springer, pp. 493--550.

        \bibitem[{Levine and Stange(2012)}]{LeSt12a}
        Levine, L., Stange, K.~E., 2012. How to make the most of a shared meal: Plan
          the last bite first. The American Mathematical Monthly 119~(7), 550--565.

        \end{thebibliography}

\end{document}